\documentclass[10pt, conference]{IEEEtran}
\usepackage[utf8]{inputenc}
\usepackage{amsmath}
\usepackage{mathtools}
\usepackage{amsfonts}
\usepackage{amssymb}
\usepackage{amsthm}
\usepackage{subfigure}
\usepackage{cite}
\usepackage{calc}
\usepackage{color}
\usepackage{epsfig}
\usepackage{setspace}
\usepackage{pstricks}
\usepackage{cancel}
\usepackage{multirow}
\usepackage{fixltx2e}
\usepackage{multicol}
\usepackage{graphicx}
\usepackage{epsfig}
\usepackage{stfloats}
\usepackage{bbm}
\usepackage{enumitem}

\DeclarePairedDelimiter{\floor}{\lfloor}{\rfloor}

\newtheorem{defn}{Definition}
\newtheorem{thm}{{\cal T}heorem}[section]
\newtheorem{cor}[thm]{Corollary}
\newtheorem{prop}{Proposition}
\newtheorem{lem}[thm]{Lemma}
\newtheorem{conj}[thm]{Conjecture}
\newtheorem{constr}[thm]{Construction}
\newtheorem{note}{Remark}

\newcommand{\bit}{\begin{itemize}}
\newcommand{\eit}{\end{itemize}}
\newcommand{\bcor}{\begin{cor}}
\newcommand{\ecor}{\end{cor}}
\newcommand{\beq}{\begin{equation}}
\newcommand{\eeq}{\end{equation}}
\newcommand{\beqn}{\begin{equation*}}
\newcommand{\eeqn}{\end{equation*}}
\newcommand{\bea}{\begin{eqnarray}}
\newcommand{\eea}{\end{eqnarray}}
\newcommand{\bean}{\begin{eqnarray*}}
\newcommand{\eean}{\end{eqnarray*}}
\newcommand{\ben}{\begin{enumerate}}
\newcommand{\een}{\end{enumerate}}
\newcommand{\bdefn}{\begin{defn}}
\newcommand{\edefn}{\end{defn}}
\newcommand{\bnote}{\begin{note}}
\newcommand{\enote}{\end{note}}
\newcommand{\bprop}{\begin{prop}}
\newcommand{\eprop}{\end{prop}}
\newcommand{\blem}{\begin{lem}}
\newcommand{\elem}{\end{lem}}
\newcommand{\bthm}{\begin{thm}}
\newcommand{\ethm}{\end{thm}}
\newcommand{\bconj}{\begin{conj}}
\newcommand{\econj}{\end{conj}}
\newcommand{\bconstr}{\begin{constr}}
\newcommand{\econstr}{\end{constr}}
\newcommand{\bpf}{\begin{proof}}
\newcommand{\epf}{\end{proof}}

\begin{document}
\title{Rack-Aware Cooperative Regenerating Codes}
\author{
\IEEEauthorblockN{Shreya Gupta, V. Lalitha\\}
\IEEEauthorblockA{SPCRC, International Institute of Information Technology Hyderabad\\
Email: \{shreya.gupta@research.iiit.ac.in,  lalitha.v@iiit.ac.in\}\\}
\vspace{-0.3cm}
}

\maketitle

 \begin{abstract}
 In distributed storage systems, cooperative regenerating codes tradeoff storage for repair bandwidth in the case of multiple node failures. In rack-aware distributed storage systems, there is no cost associated with transferring symbols within a rack. Hence, the repair bandwidth will only take into account cross-rack transfer. Rack-aware regenerating codes for the case of single node failures have been studied and their repair bandwidth tradeoff characterized. In this paper, we consider the framework of rack-aware cooperative regenerating codes for the case of multiple node failures where the node failures are uniformly distributed among a certain number of racks. We characterize the storage repair-bandwidth tradeoff as well as derive the minimum storage and minimum repair bandwidth points of the tradeoff. We also provide constructions of minimum bandwidth rack-aware cooperative regenerating codes for all parameters.
  \end{abstract}

\section{Introduction}

Traditional erasure coding techniques, though storage-efficient,  incur a large cost in terms of repair bandwidth to handle node failures where repair bandwidth is the total amount of data download needed
to repair a failed node. Regenerating codes \cite{dimakis2010network} are a class of codes designed to efficiently tradeoff storage efficiency for repair bandwidth. In the following, we will describe two variants of the setting of regenerating codes which are important in the context of the current paper. The first one allows to handle multiple node failures and are known as cooperative regenerating codes. The second one allows rack-based topology of nodes and are known as rack-aware regenerating codes.

\vspace{-0.07cm}
\subsection{Cooperative Regenerating Codes}
Cooperative Regenerating Codes are defined using parameters $(n,k,d,e,\alpha,\beta_1,\beta_2)$. The total number of nodes is denoted by $n$ and $k$ denotes the number of nodes required for the data collection process. Any $k$ out of $n$ nodes are sufficient to reconstruct the whole file. Each node consists of $\alpha$ symbols. The original data file $B$ is stored on $n$ nodes and is encoded into $n\alpha$ symbols. The number of simultaneous failures is denoted by $e$. In case of failures, cooperative regenerating codes do the repair in two rounds. In the first round, all the new (or replacement) nodes download $\beta_1$ symbols from any $d$ non-failed nodes, where $k \leq d \leq n-e$ and in the second round all the replacement nodes share their information with the all the other $e-1$ replacement nodes. So every replacement node receives $\beta_2$ symbols from $e-1$ other replacement nodes. This gives the repair bandwidth per replacement node to be $\gamma = d\beta_1 + (e-1)\beta_2$.
\par

For fixed values of $(n,k,d,e)$, there exist a bound on the file size $B$ and for the fixed value of $B$, it gives a storage$(\alpha)$-bandwidth$(\gamma)$ tradeoff. The corner points of the tradeoff where $\alpha$ is minimum, known as Minimum Storage Cooperative Regeneration(MSCR) point, and where $\gamma$ is minimum, known as Minimum Bandwidth Cooperative Regeneration(MBCR) point, are given by:
\begin{equation} \label{eq:mscr}
    (\alpha_{MSCR}, \gamma_{MSCR}) =  (\frac{B}{k}, \frac{B(d+e-1)}{d+e-k})
\end{equation}
\begin{equation} \label{eq:mbcr}
    \alpha_{MBCR} = \gamma_{MBCR} = \frac{B(2d+e-1)}{k(2d+e-k)}
\end{equation}

Cooperative regenerating codes reduce to regenerating codes by setting $e=1$.
MSCR and MBCR points become minimum storage regeneration (MSR) and minimum bandwidth regeneration (MBR) points, respectively, when $e=1$. Cooperative regenerating codes have been studied in \cite{KermarrecScouarnec,Shum13}. Explicit MBCR codes have been constructed in \cite{wang2013exact,shum2016cooperative}. Constructions of MSCR codes have been presented in \cite{shum2016cooperative,YeBargMSCR}.

There is another setting of regenerating codes for the case of multiple node failures where a central node downloads all the information from the helper nodes and reconstructs all the failed nodes. This setting is known as centralized repair and has been studied in \cite{zorgui2019centralized}.

\subsection{Rack-aware Regenerating Codes (RRCs)}\label{sec:rrgc}

Rack-Aware Regenerating code (RRC) is denoted by parameters $n, k, d, r, \alpha, \beta$. In RRC, a data file of size $B$ is encoded into $n\alpha$ symbols and stored on $n$ nodes in $r$ racks where $r$ divides $n$ such that each rack consists of $\frac{n}{r}$ nodes. The data collector reconstructs the whole file and requires downloading $\alpha$ symbols from any $k$ out of $n$ nodes. In case of a failure of a node in rack $h$, the new node downloads $\beta$ symbols from any $\floor{\frac{kr}{n}} \leq d < r$ helper racks(or relayers) other than rack $h$ and $\alpha$ symbols from all the other nodes within the rack $h$ to reconstruct the failed node. Each rack contains a relayer node which can access the contents of other nodes within the rack. The cross-rack repair bandwidth in RRC to reconstruct the failed node is given by: $\gamma = d\beta$. \par
For fixed parameters $(n, k, d, r)$ and $B$, there exist a storage$(\alpha)$ - bandwidth$(d\beta)$ tradeoff. The point on the tradeoff curve where storage per node $\alpha$ is minimum is denoted by Minimum Storage Rack-aware Regeneration (MSRR) point and the point where cross rack repair bandwidth is minimum is denoted by Minimum Bandwidth Rack-aware Regeneration (MBRR) point. The parameters of the MSRR and MBRR point are given below:
\begin{equation} \label{eq:msrr}
    (\alpha_{MSRR}, \gamma_{MSRR}) = (\frac{B}{k}, \frac{Bd}{k(d-m+1)}),
\end{equation}
\begin{equation} \label{eq:mbrr}
    \alpha_{MBRR} = \gamma_{MBRR} = \frac{Bd}{(k-m)d+m(d-\frac{m-1}{2})}.
\end{equation}

The storage repair bandwidth tradeoff and the constructions of rack-aware regenerating codes have been investigated in \cite{hou2019rack}.
A related version of storage-repair bandwidth tradeoff has been studied in the context of clustered distributed storage systems in \cite{prakash2018storage}. The tradeoff for the case of multiple node failures where all the node failures are in a single cluster have been dealt with in \cite{abdrashitov2017storage}.

\subsection{Our Contributions}
In this paper, we consider the framework of rack-aware cooperative regenerating codes for the case of multiple node failures where the node failures are uniformly distributed among a certain number of racks. We characterize the storage repair-bandwidth tradeoff as well as derive the minimum storage and minimum repair bandwidth points of the tradeoff. We also provide constructions of minimum bandwidth rack-aware cooperative regenerating codes for all parameters. The system model and the tradeoff are introduced in Section \ref{sec:sys_model}. The construction of minimum bandwidth rack-aware cooperative regenerating codes is presented in Section \ref{sec:mbrcr}.

\section{System Model}\label{sec:sys_model}
The problem of multi-node repair in case of rack-aware regenerating codes is characterized by the parameters ($B$, $n$, $k$, $d$, $r$, $e$, $f$, $\alpha$, $\beta_1$, $\beta_2$). We consider a distributed storage system consisting of $n$ nodes which are equally divided into $r$ racks and store $B$ amount of information. In this paper we assume that $r$ divides $n$. So, each rack contains $\frac{n}{r}$ nodes. We define $X\textsubscript{h,i}$ as the $i$-th node in rack $h$ where $i$ $=$ $1$,$2$,$\ldots$,$\frac{n}{r}$ and $h$ $=$ $1, 2,\ldots,r$. Each node consists of $\alpha$ symbols. A $data$ $file$ of size $B$ is encoded into $n\alpha$ symbols on $n$ nodes. \par
Each rack has a distinguished node called the $relayer$ $node$ which can access the contents of other nodes within the same rack. The system should satisfy the following two properties$\colon$

\begin{itemize}
    \item Reconstruction property$\colon$ any $k\leq n$ nodes out of $n$ nodes should be able to reconstruct the whole file.
    \item Regeneration Property$\colon$ We consider failure of $e$ nodes in any $f$ racks such that $f$ divides $e$ and each of the $f$ racks has $\frac{e}{f}$ failed nodes, we consider this case for simplicity so that cross rack bandwidth is uniform. There are $d$ helper racks (or relayers) where $m \leq d\leq r-f$ and $m=\floor*{\frac{kr}{n}}$. We use cooperative repair across racks and centralized repair within the racks. There are two rounds of repair. In the first round, each of the $f$ racks, which have failed nodes, downloads $\beta_1$ symbols from each of the $d$ helper racks. In the second round, all $f$ racks share information with each other. So every rack which has failed nodes downloads $\beta_2$ symbols from each of the other racks which have failed node. Thus, the cross rack repair bandwidth for one rack in case of $e$ node failures  $\colon$
    $\gamma=$ $d \beta\textsubscript{1}$ $+$ $(f-1) \beta\textsubscript{2}$. 
    
\end{itemize}


The aim of this paper is to characterize the tradeoff between the storage per node $\alpha$ and repair bandwidth $\gamma$. In this paper, we derive the tradeoff for functional repair. We call an encoding scheme which satisfies the above requirements with parameters $n, k, d, r, e, f, \alpha, \beta_1, \beta_2$ as a $rack-aware$ $cooperative$ $storage$ $system$ (RCSS).

\subsection{Information Flow Graphs}
The storage system described in the previous section can be represented by $information$ $flow$ $graphs$ (IFGs). Our IFG contains a vertex $S$ which represents the original data file and a vertex $T$ which represents the data collector. 


We have vertices $\mathsf{Out\textsubscript{h,i}}$ for $i$-th node in rack $h$ where $h=1,2,...,r$ and $i=1,2,...,n/r$. Edges from $S$ to $\mathsf{Out\textsubscript{h,i}}$ have capacity $\alpha$. First node ($X\textsubscript{h,1}$) in rack $h$ for $h=1,2,...,r$ is considered as the relayer node and can access the data from other nodes within the same rack. For each $h=1,2,...,r$ and $i=2,3,...,\frac{n}{r}$, there are edges of infinite capacity from each $\mathsf{Out\textsubscript{h,i}}$ to $\mathsf{Out\textsubscript{h,1}}$. 

\begin{figure}[h!]
\centering
\includegraphics[width=3.55in]{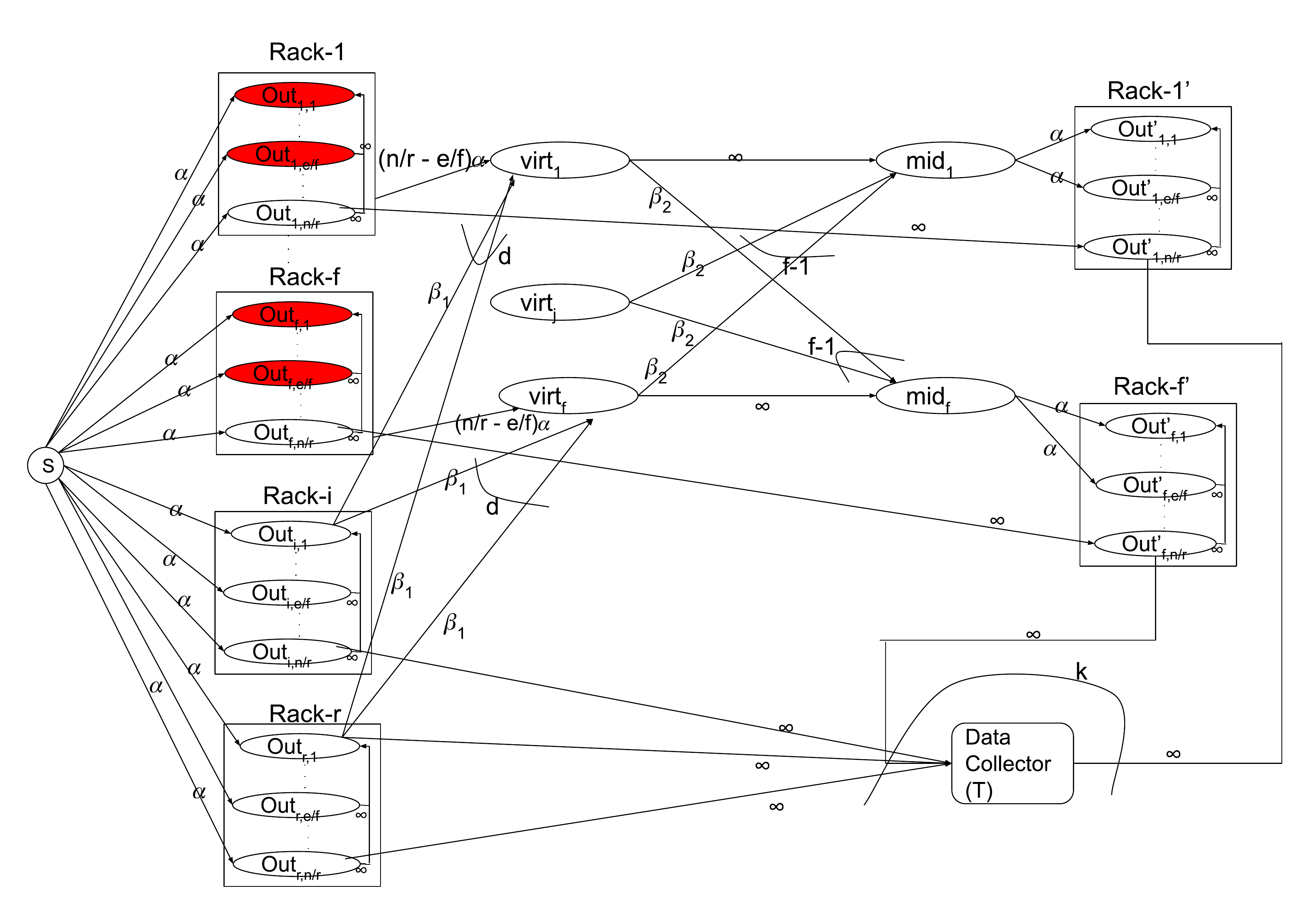}
\caption{Information Flow Graph of Rack-aware Cooperative Storage System. Red nodes denote the failed nodes.}
\label{fig:rcss}
\end{figure}

We consider $e$ nodes fail in $f$ racks with $\frac{e}{f}$ failures in each of the $f$ racks. We define a stage as the point of recovery of $e$ nodes in $f$ racks, failed in a previous stage and then simultaneous failure of $e$ nodes in $f$ racks. At $s=0$, first set of $e$ nodes fail. For $s=1,2,3,...,$ let $F\textsubscript{s}$ be the set of $f$ racks which have nodes that failed in stage $s-1$ and are regenerating in stage $s$. The set $F\textsubscript{s}$ $\subseteq$ $\{ 1,2,...,r \} $ and $|F\textsubscript{s}|=f$. $F\textsubscript{s}$ is also known as $repair$ $group$ which is a set of racks with failed nodes that get reconstructed simultaneously in stage $s$. For each rack $h$ in $F\textsubscript{s}$, we construct vertex $virt\textsubscript{h}$, vertex $mid\textsubscript{h}$ and vertices corresponding to all the nodes within that rack denoted by $\mathsf{Out^{\text{'}}_{\text{h,i}}}$ for $i=1,2,...,n/r$. For the nodes in rack $h$ which have not failed in stage $s-1$, there is an edge of capacity infinity from $\mathsf{Out\textsubscript{h,i}}$ to $\mathsf{Out^{\text{'}}_{\text{h,i}}}$. To reconstruct the failed nodes in rack $h$, vertex $virt_h$ has $d$ incoming edges with capacity $\beta_1$ from $d$ helper racks and $(\frac{n}{r} - \frac{e}{f})$ incoming edges with capacity $\alpha$ from the nodes which have not failed in stage $s-1$ within rack $h$. Vertex $virt_h$ is connected to vertex $mid_h$ with an edge of capacity infinite. For $h, h^{'} \in F_s$, $h \neq h^{'}$, we join $virt_h$ to $mid\textsubscript{h\textsuperscript{'}}$ with capacity $\beta_2$. $virt_h$ and $mid_h$ in information flow graph together work as $cent_h$ which is the node where centralized reconstruction happen for the failed nodes within rack $h$. There are directed edges of capacity $\alpha$ from $mid_h$ to $\mathsf{Out^{\text{'}}_{\text{h,i}}}$ for $i$ belonging to the failed nodes in rack $h$ in stage $s-1$. The download at $virt_h$ where the replacement rack downloads $\beta_1$ symbols from $d$ other racks and $\alpha$ symbols from the remaining nodes within the same rack corresponds to the first round of repair and download at $mid_h$ where each replacement rack downloads $\beta_2$ symbols from every other rack corresponds to the second round of repair. We also have an edge with infinite capacity from each $\mathsf{Out^{\text{'}}_{\text{h,j}}}$ to $\mathsf{Out^{\text{'}}_{\text{h,1}}}$ for $j=2,...,\frac{n}{r}$. Data collector $T$ is connected to any $k$ out of $n$ existing nodes with capacity of infinite.

Given parameters $n,k,d,r,e,f,\alpha,\beta_1,\beta_2$, there can be many IFGs based on the failure pattern. The set of all such IFGs is denoted by  $\mathbb{G}(n,k,d,r,e,f,\alpha, \beta_1, \beta_2)$. Given an IFG $G$ $\in$ $\mathbb{G}$, there can be many different data collectors connecting to any $k$ data nodes out of $n$. We denote the set of all data collectors corresponding to IFG $G$ by $DC(G)$. For an IFG $G$ $\in$ $\mathbb{G}$ with source vertex $S$ and a data collector $T$ $\in$ $DC(G)$, the $(S,T)$-cut is defined as the sum of capacities of a subset of edges of graph $G$ which when removed from the graph, partition the vertices of $G$ such that $S$ and $T$ are disconnected. The smallest capacity of an $(S,T)$-cut in a given IFG $G$ is denoted by $mincut\textsubscript{G}(S,T)$. According to the max-flow bound in network thoery, the supported file size is upper bounded by $mincut\textsubscript{G}(S,T)$ minimized over all data collectors $T \in DC(G)$ and IFG $G$ $\in$ $\mathbb{G}$.

The following lemma is very similar to the one in \cite{hou2019rack} and it has been discussed here for the sake of completeness.
\begin{lem} \label{lem:relayer}
If a relayer node in a rack is connected to the data collector $T$ and not all the other remaining $\frac{n}{r}-1$ nodes are connected to $T$, then the capacity of $(S,T)$-cut is not minimum.
\end{lem}

\begin{proof}
Suppose a relayer $X_{h,1}$ for any $h = 1,2,...,r$ is connected to the data collector $T$. All the incoming edges of $T$ have infinite capacity. Assuming that rack $h$ does not have any failed nodes, we consider only the incoming edges of $\mathsf{Out_{h,i}}$ for $i=1,2,...,n/r$, each with capacity $\alpha$, which in total contribute $\alpha \frac{n}{r}$ to the cut. So, relayer node $X_{h,1}$ contribute $\alpha \frac{n}{r}$ to the cut in case of no node failure within the rack.  On the other hand, if rack $h$ has failed and then reconstructed nodes and the relayer node $X_{h,1}^{'}$ is connected to the $T$, then we can have cut of capacity $\alpha \frac{n}{r}$ considering the incoming edges of $\mathsf{Out_{h,i}^{'}}$  and $\mathsf{Out_{h,i}}$ for the cut or of capacity $d\beta_1 + (\frac{n}{r} - \frac{e}{f})\alpha$ considering the incoming edges of $virt_h$ and $(f-1)\beta_2$ considering the incoming edges of $mid_h$  for the cut. So a failed relayer node contributes $min(\alpha \frac{n}{r}, (d\beta_1 + (\frac{n}{r} - \frac{e}{f})\alpha + (f-1)\beta_2))$ to the cut. In both cases of failed relayer or non failed relayer, if a relayer is connected to $T$ then the other nodes within the rack if connected or not connected to $T$, will not have any contribution to the cut. So  to minimize the capacity of the cut, if a relayer is connected to $T$ then all the other nodes within the rack should also be connected to $T$.
 \end{proof}

\begin{thm}
For fixed parameters $n, k, d, r, e, f, \alpha, \beta_1, \beta_2$, where $m \leq d \leq (r-f)$ and $f | m$, if  there exist an RCSS with file size $B$, then it will satisfy $\colon$

\begin{equation} \label{eq:tradeoff}
    B \leq k\alpha + \sum_{i=1}^{g} u_i \min(0, (d-\sum_{j=1}^{i-1} u_j)\beta_1 - \frac{e}{f}\alpha + (f-u_i)\beta_2)
\end{equation}


where


   $\textbf{u}=[u_1, u_2, ..., u_g]$ and $1 \leq u_i \leq f, g \in \mathbb{N}, \sum_{i=1}^{g}u_i=m.$
   \end{thm}

\begin{proof}
We consider $g$ number of stages (or repair groups of size $f$) and $u_i$ denotes number of racks contacted in stage (or from repair group) $i$ for the data collection process. Let $I$ denote the set of nodes which will contribute for the data collection process. 
We want to find the min-cut to get the upper bound on file size. To partition IFG $G$ into sets $U$ and $\Bar{U}$, we do not consider edges with infinite capacity. So, $S \in U$ and either $\mathsf{Out^{\text{'}}_{\text{h,i}}} \in \Bar{U}$ or $\mathsf{Out_{\text{h,i}}} \in \Bar{U}$ for $(h,i) \in I$. We know that IFGs are Directed acyclic graphs (DAG) and every DAG has a topological sorting. As we know that all the racks in a repair group have nodes which are reconstructed simultaneously so all the nodes in these racks will be adjacent to each other in topological sorting. If we consider topological sorting for the nodes connected to the data collector, then nodes in rack in $i$-th repair group do not have incoming edges from nodes in rack in $j$-th repair group for $j > i$ where $i,  j \in [g]$. 



    
From Lemma \ref{lem:relayer}, we can see that a relayer node contributes $min(\frac{n}{r}\alpha, d\beta_1+(\frac{n}{r}-\frac{e}{f})\alpha+(f-1)\beta_2)$ to the cut. Including a relayer node means including the whole rack for the data collection process. We take $u_i$ racks from the $i$-th repair group for data collection. For the reconstruction of a rack in $i$-th repair group, at most $\sum_{j=1}^{i-1} u_j$ edges come from the racks which are reconstructed in the previous steps and are already included for the data collection process.
Thus, $i$-th repair group contributes at least $u_imin(\frac{n}{r}\alpha, (d-\sum_{j=1}^{i-1}u_j)\beta_1 + (\frac{n}{r}-\frac{e}{f})\alpha + (f-u_i)\beta_2 ))$ to the cut. Finally, adding all the contributions from different repair groups, we get contributions for $m\frac{n}{r}$ nodes as $\sum_{i=1}^{g}u_i = m$ and each of the remaining $k - m\frac{n}{r}$ nodes have $\alpha$ contribution to the cut. This gives us the min-cut for any IFG $G(n,k,d,r,e,f,\alpha,\beta_1, \beta_2)$ and thus the upper bound on file size.

Next, we show that there exist an IFG where min-cut is equal to the right hand side value in \eqref{eq:tradeoff}. We consider the racks $1,2,..m$ participating in data collection process and hence, first $u_{1}$ racks taken for data collection in first stage (considering they are reconstructed in that stage) and next $u_{2}$ in next stage and so on, until last $u_{g}$ racks in $g$-th stage, considering that each of them have failed nodes, so each of them will contribute $u_{i}min(\alpha\frac{n}{r}, ((d-\sum_{j=1}^{i-1}u_{j})\beta_1 + (\frac{n}{r}-\frac{e}{f})\alpha + (f-u_i)\beta_2))$ and rest $k-m\frac{n}{r}$ nodes from the $(m+1)$-th rack will contribute $\alpha$ to the cut. Summing all the values will give min-cut which is right side in \eqref{eq:tradeoff}.
\end{proof}

\subsection{Optimal Tradeoffs} 
We would like to find points where storage cost $\alpha$ is minimized and repair bandwidth $\gamma$ is minimized under the constraints of \eqref{eq:tradeoff}.

$1)$ {\bf MSRCR Point:} Minimum Storage Rack-aware Cooperative Regeneration Codes are optimal codes which provide the lowest possible storage cost $\alpha$ while minimizing the repair bandwidth $\gamma$. 
\begin{equation*}
\alpha = \frac{B}{k}, \beta_1 = \frac{B}{k}\frac{e}{f}\frac{1}{d-m+f}, \beta_2 = \frac{B}{k}\frac{e}{f}\frac{1}{d-m+f}
\end{equation*}

These values are determined in two steps. In first step, we consider two particular cuts to find the minimum values of parameters $\alpha$, $\beta_1$ and $\beta_2$ which ensure that the max flow is at least equal to the file size $B$ which proves the optimality of the solution if correct. The correctness of these parameters is proved by showing that they are sufficient for all possible cuts.

$Proof$ $of$ $MSRCR$ $Optimality\colon$ First we minimize $\alpha$ and then for the minimum value of $\alpha$ we minimize repair bandwidth $\gamma$. It is clear from the \eqref{eq:tradeoff} that we get $\alpha = \frac{B}{k}$ as the minimum value of $\alpha$. Now we consider two particular cases of repairs $(\textbf{u}=[1,1,...,])$ and $(\textbf{u}=[f,f,...,])$ to minimize repair bandwidth.\\

$Case$ $1\colon$ When $u_i=f$, $\forall i \in \{1,2,...,g\}$, then we require that
$0 \leq (d-\sum_{j=1}^{i-1}f)\beta_1 - \frac{e}{f}\alpha, \forall i \in \{1,2,...,\frac{m}{f}\},$
which leads to $\beta_1  \geq \frac{B}{k}\frac{e}{f}\frac{1}{d-m+f}$.

$Case$ $2\colon$ When $u_i=1$, $\forall i \in \{1,2,...,g\}$, then we want
    $0 \leq (d-\sum_{j=1}^{i-1}1)\beta_1 - \frac{e}{f}\alpha + (f-1)\beta_2, \forall i \in \{1,2,...,m\},$
which results in
\begin{equation} \label{eq:beta2_min1}
    \beta_2 \geq \frac{1}{f-1}(\frac{e}{f}\alpha - (d-m+1)\beta_1).
\end{equation}

Substituting the minimum value of $\beta_2$ from above in the expression of cross rack repair bandwidth for one rack $\gamma = d\beta_1 + (f-1)\beta_2$, we have $\gamma$ in terms of $\beta_1$ as follows: \\
    $\gamma = \frac{e}{f}\frac{B}{k} + (m-1)\beta_1.$
 
This shows that the repair cost increases linearly with $\beta_1$, so to minimize $\gamma$ we need to minimize $\beta_1$. We know from Case 1 that the minimum value of 
    $\beta_1 = \frac{B}{k}\frac{e}{f}\frac{1}{d-m+f}.$ 

We can get the corresponding value of $\beta_2$ from \eqref{eq:beta2_min1}. We can observe that $\beta_1 = \beta_2$ for MSRCR codes.\\

%
%
%
%
%
%

$2)$ {\bf MBRCR Point:} Minimum Bandwidth Rack-aware Cooperative Regeneration Codes are optimal codes which provide the lowest possible repair bandwidth$(\gamma)$ while minimizing the storage cost $\alpha$. The $\alpha$, $\beta_1$ and $\beta_2$ parameters for MBRCR point will be derived to be :
\begin{equation*}
\alpha = \frac{f}{e}\gamma, \beta_1 = \frac{B}{k\frac{f}{e}(d+\frac{f-1}{2}) + \frac{m-m^2}{2}} 
\end{equation*}
\begin{equation*}
    \beta_2 = \frac{1}{2}\frac{B}{k\frac{f}{e}(d+\frac{f-1}{2}) + \frac{m-m^2}{2}}.
\end{equation*}
$Proof$ $of$ $MBRCR$ $Optimality \colon$ For MBRCR codes, we want to minimize $\gamma$ before $\alpha$. From the upper bound of file size, we can say that
    $0 \geq d\beta_1 + (\frac{n}{r}-\frac{e}{f})\alpha + (f-1)\beta_2,$
which implies that
$\alpha \geq \frac{f}{e}\gamma$.

Let us take two particular 
cases of repairs $(\textbf{u}=[1,1,...,])$ and $(\textbf{u}=[f,f,...,])$ to minimize repair bandwidth. \par
$Case$ $1 \colon$  When $u_i=f$, $\forall i \in {1,2,...,\frac{m}{f}}$, then
  $B \leq k\alpha + \sum_{i=1}^{\frac{m}{f}}f((d-\sum_{j=1}^{i-1}f)\beta_1 -\frac{e}{f}\alpha) ,$
which leads to   $ \beta_1 \geq \frac{B-(k-m\frac{e}{f})\alpha}{m(d-\frac{m-f}{2})}$.

$Case$ $2\colon$ When $u_i=1$, $\forall i \in {1,2,...,m}$, we have,

    $B \leq k\alpha+\sum_{i=1}^{m}((d-\sum_{j=1}^{i-1}1)\beta_1 - \frac{e}{f}\alpha + (f-1)\beta_2)$
which leads to

\begin{equation} \label{eq:beta2_min2}
    \beta_2 \geq \frac{1}{(f-1)} \frac{B-(k-m\frac{e}{f})\alpha-m\beta_1(d-\frac{m-1}{2})}{m}.
\end{equation}

Substituting the minimum value of $\beta_2$ from above in the expression of cross rack repair bandwidth for one rack $\gamma = d\beta_1 + (f-1)\beta_2$, we have $\gamma$ in terms of $\beta_1$ as follows:

    $\gamma = \frac{1}{1+\frac{f}{me}(k-m\frac{e}{f})} (\frac{B}{m} + \frac{m-1}{2}\beta_1).$ 

The above equation shows that $\gamma$ grows linearly with $\beta_1$. Hence, in order to minimize $\gamma$, we need to minimize $\beta_1$ and hence minimum value of $\beta_1$ is :
\begin{equation} \label{eq:beta1_min}
    \beta_1 = \frac{B-(k-m\frac{e}{f})\alpha}{m(d-\frac{m-f}{2})}.
\end{equation}
Substituting this value of $\beta_1$ in \eqref{eq:beta2_min2}, we get $\beta_1 = 2\beta_2$.
Substituting $\alpha = \gamma \frac{f}{e}$ and $\gamma = d\beta_1 + (f-1)\beta_2$ in \eqref{eq:beta1_min}, we get $\beta_1 = \frac{B}{k\frac{f}{e}(d+\frac{f-1}{2}) + \frac{m-m^2}{2}}.$\\

%
%
%
%
%

Correctness of MBRCR and MSRCR points can be proved by showing that these values make sure that enough information flows through every cut in any case. The proof of correctness for MBRCR and MSRCR follows along the lines of proof of correctness for MBCR and MSCR in  \cite{KermarrecScouarnec}. The proofs of correctness have been omitted because of lack of space.

\begin{note}
MSRCR and MBRCR points with $n=r$, $e=f$ coincide with MSCR and MBCR points given in \eqref{eq:mscr} and \eqref{eq:mbcr}. Similarly, MSRCR and MBRCR points with $e=f=1$ coincide with MSRR and MBRR points given in \eqref{eq:msrr} and \eqref{eq:mbrr}. \\
\end{note}

\section{Construction of MBRCR Codes for all parameters} \label{sec:mbrcr}

In this section, we will present our construction of MBRCR codes for all parameters ($B$, $n$, $k$, $d$, $r$, $e$, $f$). The idea of product-matrix construction has been introduced in \cite{rashmi2011optimal} in the context of regenerating codes. We generalize the construction of MBRR construction given in \cite{hou2019rack} to the case of multiple erasures using the product-matrix construction of MBCR code given in \cite{shum2016cooperative}.
%

\begin{constr}\label{constr:mbcr_local}
\normalfont
 The MBRCR construction will have the parameters $ \beta_1 = \frac{2e}{f}, \beta_2 = \frac{e}{f}, \alpha = 2d + f-1$ and
\begin{eqnarray}
  B &= & k(2d+f-1) + \frac{e}{f} (m-m^2), \\
 \hspace{-0.2in} & = & \left ( k - m \frac{e}{f} \right ) (2d + f-1) + \frac{e}{f} \left ( m (2d + f-m) \right ). \nonumber
\end{eqnarray}
where $m = \lfloor \frac{kr}{n} \rfloor$.


 We will describe the  MBRCR code construction in the following steps.

{\bf Step-1} Generation of global coded symbols: Let $[s_1, s_2, \ldots, s_B]$ denote $B$ message symbols over $\mathbb{F}_q$. Encode these message symbols using the generator matrix $G$ of a MDS code where $G$ is of size $B \times (n - r\frac{e}{f})(2d + f -1) + \frac{e}{f} ( m (2d + f-m))$.

{\bf Step-2} Filling racks partially with global coded symbols: Divide the $(n - r\frac{e}{f})(2d + f -1)$ of the global coded symbols into $r$ groups of $(\frac{n}{r} - \frac{e}{f})(2d + f -1)$ symbols each. These symbols are then used to fill the last $\frac{n}{r} - \frac{e}{f}$ nodes in each rack.

{\bf Step-3} Generating MBCR code symbols: Consider the remaining $\frac{e}{f} ( m (2d + f-m))$ of the global coded symbols which were not filled in the step above. We will generate MBCR coded symbols corresponding to $\frac{e}{f}$ MBCR codes as follows:

Consider a message matrix $M_i, 1 \leq i \leq \frac{e}{f}$ formed of $m(2d + f -m)$ global coded symbols as follows:
\begin{equation}\label{eq:message_matrix}
M_i \triangleq \left [ \begin{array}{c|c}
A_i&B_i \\
\hline 
C_i&0
\end{array} \right ]
\end{equation}
where $A_i$ is a matrix of size $m \times m$, $B_i$ is a matrix of size $m \times (d+f-m)$ and $C_i$ is a matrix of size $(d-m) \times m$. Let $U$ and $V$ denote matrices of sizes $d \times r$ and $(d + f) \times r$. $U = \begin{bmatrix} U^{(1)}_{m \times r} \\ U^{(2)}_{(d-m) \times r} \end{bmatrix}$ such that any $m \times m$ submatrix of $U^{(1)}$ is full rank and any $d \times d$ submatrix of $U$ is full rank. $V = \begin{bmatrix} V^{(1)}_{m \times r} \\ V^{(2)}_{(d-m) \times r} \end{bmatrix}$ such that any $m \times m$ submatrix of $V^{(1)}$ is full rank and any $(d+f) \times (d+f)$ submatrix of $V$ is full rank. $\bold{u}_l$ denote the $l^{\text{th}}$ column of $U$ and $\bold{v}_l$ denote the $l^{\text{th}}$ column of $V$. In the  $i^{\text{th}}$ node of $l^{\text{th}}$ rack ($1 \leq i \leq \frac{e}{f}, 1 \leq l \leq r$), $\alpha = 2d + f-1$ symbols will be calculated based on the following set of symbols: $M_i  \bold{v}_l$ as well as $M_i^T  \bold{u}_l$. There is a linear dependence relation among the above $2d+f$ code symbols given by $ \bold{u}_l^T  M_i  \bold{v}_l = \bold{v}_l^T M_i^T  \bold{u}_l$.

{\bf Step-4} Addition of local parities to the MBCR code symbols: Let $\bold{c}_l = [\bold{c}_{1,l}, \bold{c}_{2,l}, \ldots, \bold{c}_{\frac{n}{r}-\frac{e}{f},l}]$ denote the vector of global code symbols stored in the last $\frac{n}{r} - \frac{e}{f}$ nodes of rack $l$. Let $P_{i,l}$ denote a matrix of size $(2d + f) \times (\frac{n}{r} - \frac{e}{f})(2d+f-1)$, where the last row is filled with zeros. Consider the following set of $2d+f$ symbols obtained by adding local parities to the MBCR code symbols generated above:
\begin{equation*}
\left [ \begin{array}{c} M_i  \bold{v}_l \\ M_i^T  \bold{u}_l \end{array} \right ] + P_{i,l} \bold{c}_l  .
\end{equation*}
The first $2d+f-1$ symbols of the above vector are stored in the $i^{\text{th}}$ node in the $l^{\text{th}}$ rack. Also, $P_{i,l}$ are required to satisfy the property that $\bold{c}_{1,l}, \bold{c}_{2,l}, \ldots, \bold{c}_{\frac{n}{r}-\frac{e}{f},l}, P_{1,l} \bold{c}_l, \ldots, P_{\frac{e}{f},l} \bold{c}_l$ are such that any $\frac{n}{r}-\frac{e}{f}$ of them are sufficient to recover the rest (also known as vector-MDS property).

\end{constr}

{\bf File Recovery:} It is easy to see that if the $k$ nodes are picked from the last $ \frac{n}{r} - \frac{e}{f}$ nodes of each rack, then since the global coded symbols are formed based on MDS code, we can recover all the $B$ message symbols. If the $k$ nodes are picked partly from the first $\frac{e}{f}$ nodes in each rack and the rest from $ \frac{n}{r} - \frac{e}{f}$ nodes, it can be argued (using argument similar to those in the Proof of Theorem 8 in \cite{hou2019rack})  that by picking elements of $U$, $V$ and $\{P_{i,l}\}$ from a sufficiently large field, we can ensure that the encoding matrix corresponding to symbols in these $k$ nodes is full rank (rank $B$). Hence we can recover at least $B$ global coded symbols and as a result we can also recover $B$ message symbols by the property of the MDS code. The details are omitted due to lack of space.

{\bf Node Repair:}  Each of the $d$ helper racks sends $2 \frac{e}{f}$ symbols $\bold{u}_l^T  M_i  \bold{v}_j$ and $\bold{v}_l^T  M_i^T  \bold{u}_j$, $1 \leq i \leq \frac{e}{f}$. Based on these $M_i  \bold{v}_l$ can be recovered at each of the failed node. In the second round of repair, $\bold{u}_t^T M_i  \bold{v}_l$ is sent from rack $l$ to rack $t$, where $l$ and $t$ are both failed racks. Thus, the failed racks can completely recover the symbols corresponding to $\left [ \begin{array}{c} M_i  \bold{v}_l \\ M_i^T  \bold{u}_l \end{array} \right ]$. Hence, the MBCR code symbols can be generated at every failed node. Assume that $\ell_1$ of the first set of $\frac{e}{f}$ nodes and $\ell_2$ of the last set of nodes fail in a particular rack such that $\ell_1 + \ell_2 = \frac{e}{f}$, then it is clear that local parities corresponding to $\frac{e}{f} - \ell_1$ can be recovered first by subtracting the MBCR code symbols from them. In that case, a total of $\frac{e}{f} - \ell_1 + (\frac{n}{r} - \frac{e}{f} - \ell_2) = \frac{n}{r} - \frac{e}{f}$ node contents (global coded symbols + local parities) are available. These can be used for recover the rest of $\frac{e}{f}$ global coded symbols + local parities (from the vector-MDS property). Hence, the nodes can be repaired.

\bibliography{biblo}
\bibliographystyle{ieeetr}

\end{document}